\documentclass[conference]{IEEEtran}
\usepackage{cite}
\usepackage[pdftex]{graphicx}
\usepackage[pdftex]{color}
\graphicspath{/}
\usepackage[ruled,vlined,linesnumbered]{algorithm2e}
\usepackage[cmex10]{amsmath}
\usepackage{algpseudocode}
\usepackage{array}
\usepackage{url}
\usepackage{bm} 
\usepackage{amssymb}
\usepackage{epstopdf}
\usepackage{amsthm}
\newtheorem{thm}{Theorem}[section]
\newtheorem{cor}[thm]{Corollary}
\newtheorem{lem}[thm]{Lemma}

\theoremstyle{remark}

\theoremstyle{definition}

\newcommand{\C}[1]{C\left( #1 \right)}

\begin{document}
\title{Network-Assisted Device-to-Device (D2D) Direct Proximity Discovery with Underlay Communication}

\author{\IEEEauthorblockN{Nuno K. Pratas, Petar Popovski}
\IEEEauthorblockA{Department of Electronic Systems, Aalborg University, Denmark \\
Email: \{nup,petarp\}@es.aau.dk}}

\maketitle

\begin{abstract}

Device-to-Device communications are expected to play an important role in current and future cellular generations, by increasing the spatial reuse of spectrum resources and enabling lower latency communication links. This paradigm has two fundamental building blocks: (i) proximity discovery and (ii) direct communication between proximate devices.
While (ii) is treated extensively in the recent literature, (i) has received relatively little attention. 
In this paper we analyze a network-assisted underlay proximity discovery protocol, where a cellular device can take the role of: \emph{announcer} (which announces its interest in establishing a D2D connection) or \emph{monitor} (which listens for the transmissions from the announcers).
Traditionally, the announcers transmit their messages over dedicated channel resources.
In contrast, inspired by recent advances on receivers with multiuser decoding capabilities, we consider the case where the announcers underlay their messages in the downlink transmissions that are directed towards the monitoring devices.
We propose a power control scheme applied to the downlink transmission, which copes with the underlay transmission via additional power expenditure, while guaranteeing both reliable downlink transmissions and underlay proximity discovery.

\end{abstract}



\section{Introduction}
\label{sec:Introduction}

Network assisted \emph{Device-to-Device} (D2D) communications, stands for the network assisted establishment of direct communication links between proximate cellular devices~\cite{6163598,DBLP:journals/corr/AsadiWM13,6807945,6231164}; and is expected to play an important role on 5G systems~\cite{Fantastic5G}.
Prior to the establishment of these links, the network needs first to become aware if these devices want to communicate and if they are in close proximity to each other.
In 3GPP -- where D2D falls within the Proximity Services (ProSe) umbrella~\cite{3GPPTS23.303} -- this information can be inferred from the complementary proximity discovery mechanisms taking place both in the Evolved Packet Core (EPC) and in the radio access network (E-UTRAN)~\cite{6163598}.
In the EPC -- denoted as EPC-level discovery -- the proximity between devices can be extracted from the device's periodic location updates, while the interest to establish direct communication between devices can either be triggered by the devices or from the network side.
The latter, can be triggered from the monitoring of the ongoing communications between devices taking place over the network infrastructure~\cite{6163598}.
On the other hand, in the E-UTRAN -- denoted as direct discovery -- the proximity is established by direct communication between the interested devices.
There are two direct proximity discovery protocols being currently specified in 3GPP~\cite{3GPPTS23.303}:
Model A ("\emph{I am here}") and model B ("\emph{Who is there?}").

In the model A discovery protocol, as illustrated in Fig.~\ref{fig:D2DDiscoveryScenario}(a), the cellular devices can take either the role of \emph{announcer} or of \emph{monitor}.
Devices in the announcer role, broadcast an announcing message with information that could be of interest to the monitors within its proximity region.
While, devices in the monitor role, listen for certain information of interest, transmitted by the announcers within their vicinity.
The announcer transmission, is dimensioned with fixed power and rate such that the devices within a \emph{proximity region} will detect it with high probability.

In the model B discovery protocol, as illustrated in Fig.~\ref{fig:D2DDiscoveryScenario}(b), the cellular devices can take either the role of \emph{discoverer} or of \emph{discoveree}.
In this model, the discoverer receives the information from the nearby discoverees.
In contrast with model A, this protocol is based on the aggregation of the transmissions from the nearby cellular devices, which are not necessarily known a priori by the network to be within the proximity region\footnote{The associated radio resource management problem is the same as in random access settings without a priori knowledge, i.e. given an unknown number of arrivals what should the number of allocated resources be?}.
\begin{figure}[tb]
    \begin{center}
        \includegraphics[width=\linewidth]{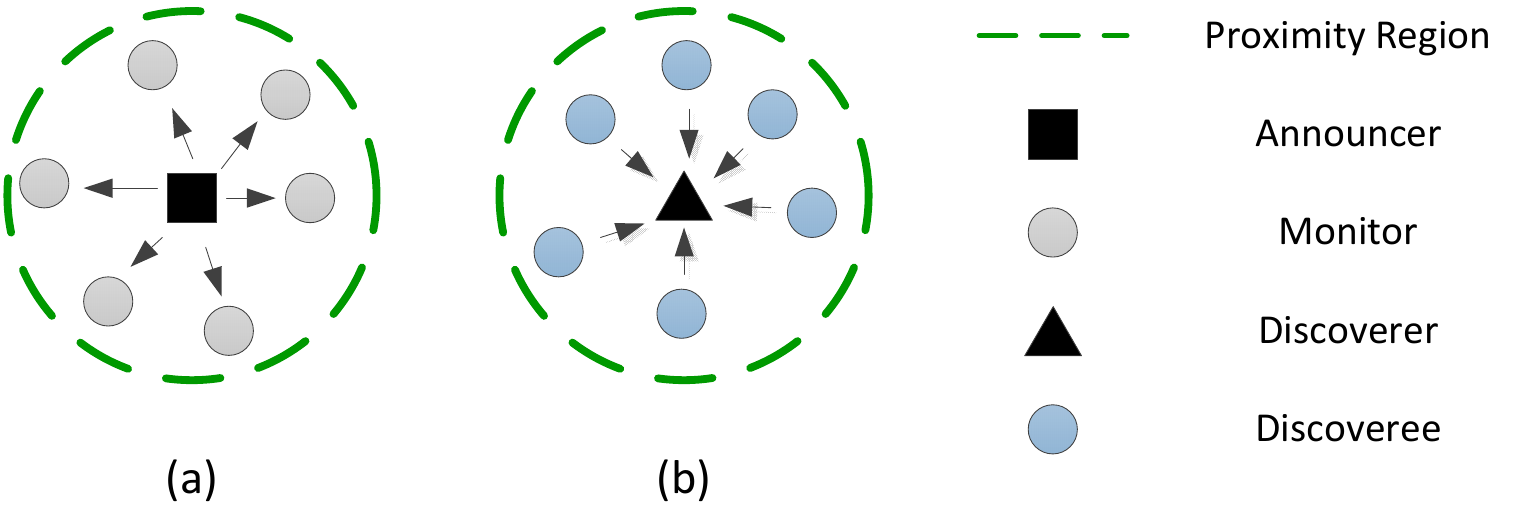}
    \end{center}
    \caption{Direct proximity discovery within a device proximity region: (a) Model A ("\emph{I am here}") and (b) Model B ("\emph{Who is there?}").}
    \label{fig:D2DDiscoveryScenario}
\end{figure}

In the remaining of the paper, we will focus our analysis on the Model A discovery protocol setting.
The radio resource management approach in place for Model A~\cite{3GPPTR36.843}, is to schedule dedicated resources to the announcer transmission, i.e. resources orthogonal to the ones assigned to the other active cellular links.
In contrast, we consider the case where the announcer transmission is scheduled to underlay the ongoing downlink transmissions to the monitors.
Therefore, creating a multiple access channel~\cite{gamal2011network}, composed by the announcer and the base station transmissions.

The occurrence of D2D underlay communication links, although previously considered in the literature~\cite{Pratas2014,PekkaJANIS2009,5350367,5073734}, have so far not been considered in a proximity discovery setting. This is due to the fact that the channel state between the announcer and monitor device is unknown, which makes the level of interference to the downlink transmissions unknown and therefore preventing the eNodeB to cope with it.
However, our recent result~\cite{7051287} shows that if the interfering signal uses a fixed rate and the downlink receiver uses multi-user decoding, then it is possible to receive the downlink transmission with zero outage, without requiring the base station to have any prior Channel State Information at the Transmitter (CSIT) of the interfering link.
The signals applied in the discovery process use fixed rate, which makes them perfect candidates for underlaying while guaranteeing zero outage to the downlink connection.
In this paper, we extend the result in~\cite{7051287}, such that a fixed rate zero-outage downlink transmission is assured at multiple monitors and provide a new formulation of the result in terms of power adaptation.
The analysis and numerical results show that the underlay operation is indeed possible, but at the cost of an increased downlink power usage that increases with the rate of the underlay transmission.

This paper is organized as follows.
In Section~\ref{sec:network_controlled_proximity_discovery}, we provide an overview of the Model A proximity discovery protocol.
We describe the system model in Section~\ref{sec:system_model}.
The analysis is provided in Section~\ref{sec:analysis}, while the numerical results are given in Section~\ref{sec:numerical_results_and_discussion}.
Finally, in Section~\ref{sec:Conclusion} we conclude the paper.

\section{Network Controlled Proximity Discovery Overview} 
\label{sec:network_controlled_proximity_discovery}

There are three network entities involved on the Model A discovery protocol, namely the announcer, the monitor and the ProSe function (located in the network infrastructure at the EPC).
The ProSe function is responsible for the admission control to ProSe, coordination of the proximity discovery and the monitoring of the ongoing D2D communication links.
In Fig.~\ref{fig:CellularControlledDiscoveryProcedure}, is depicted the four essential steps that an announcer takes to complete the proximity discovery protocol, when authorized to use ProSe.
\begin{figure}[tb]
    \begin{center}
        \includegraphics[width=\linewidth]{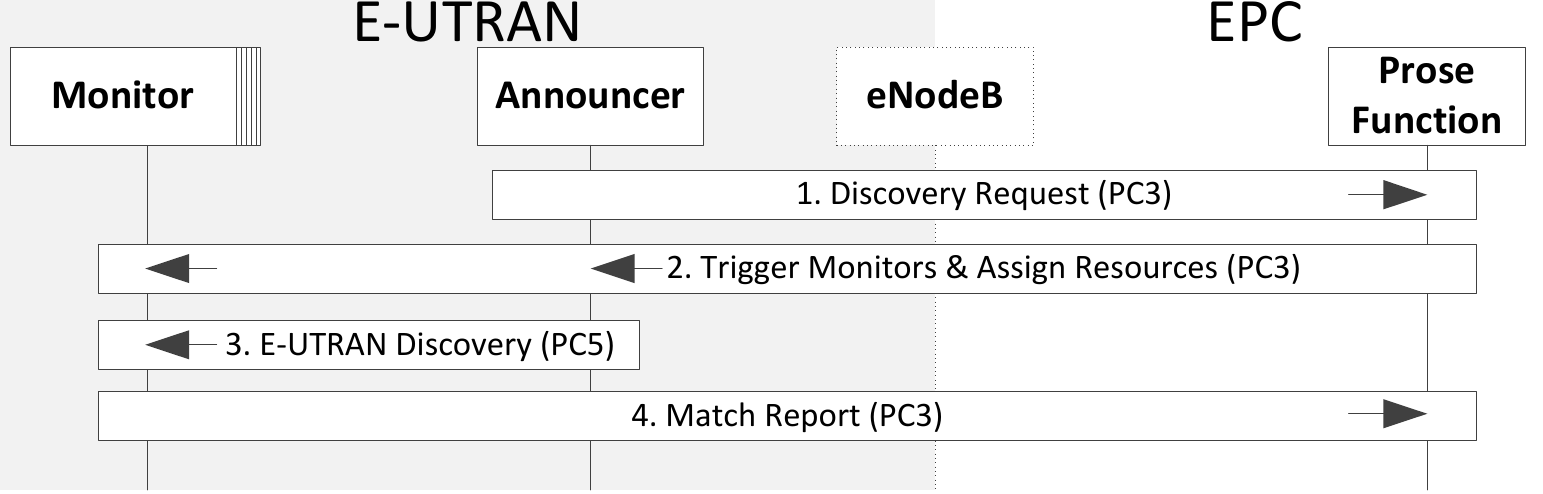}
    \end{center}
    \caption{Overview of model A network controlled ProSe direct discovery~\cite{3GPPTS23.303}. PC3 denotes the logical interface between the ProSe enabled devices and the Prose function at the Evolved Packet Core (EPC), which takes place over the E-UTRAN and EPC. PC5 denotes the logical interface between two Prose enabled devices, which takes place in the E-UTRAN.}
    \label{fig:CellularControlledDiscoveryProcedure}
\end{figure}
In the \emph{Discovery Request} step, the announcer sends a discovery request to the ProSe function.
If this request is accepted, the~\emph{Trigger Monitors \& Assign Resources} step takes place where the Prose function triggers the monitors within the announcer proximity region, assigning and informing the intervening devices (announcer and monitors) which system resources have been allocated to the transmission of the announcer message, i.e. the radio resources over which the PC5 interface will take place.
In the~\emph{E-UTRAN Discovery} step, the announcer transmits its message to the monitors, over the PC5 interface in the E-UTRAN.
Finally, in the \emph{Match Report} step, the monitors that where able to decode successfully the announcer message, transmit the received information to the ProSe function for confirmation of the announcer identity.
After this last step, the direct D2D communication link can be established.

In the remaining of the paper, we focus on the third step of the discovery protocol.
Namely, on the transmission of the announcer message underlaying the downlink transmission targeted to one of the monitors.


\section{System Model} 
\label{sec:system_model}

Here we provide first a general system model overview, followed by a detailed description of the signal and interference model.

\subsection{Overview} 
\label{sub:overview}

We focus our analysis on a single cell network composed by an eNodeB $B$ and multiple cellular devices.
A set of these devices takes the announcer role, while the remaining devices take the monitoring role, listening to the announcer message in the resources assigned for that purpose by the network.
We assume that the announcers are spread enough in space such that their proximity regions do not overlap~\footnote{In other words, the set of proximate monitors to any announcer is disjoint from the set of monitors associated with any other announcers} and that the network assigns different resources to each announcer.
The network is assumed to have knowledge about the positions of each device via periodic device location updates, while being \emph{unaware of the channel conditions between any two devices prior to the proximity discovery procedure}.
Therefore, the network knows which monitors will be potentially in the proximity region\footnote{Two devices can be spatially nearby, but due to the local radio environment characteristics they might not have a radio channel with enough quality to warrant direct device-to-device communications. The verification if the radio link has enough quality is the main purpose of the proximity discovery procedure.}. 

We assume that within the proximity region of an announcer $A$ there are $N$ nearby monitors.
We further assume that each of these monitors has a dedicated ongoing downlink transmission from $B$.
We identify the associated downlink channel and its intended receiver by the same subscript, e.g. the $i^{th}$ channel intended receiver is the $i^{th}$ device. 
We assume that these downlink transmissions have a fixed rate $R_B$, to be met with zero outage.
The fixed downlink rate assumption simplifies the introduction of our power control concept, but we note that an extension to an adaptive downlink rate setting is possible. 

Finally, we assume an OFDMA like based cellular system~\cite{3GPPTR36.843}, where: (i) the synchronization between network devices is assured by the cyclic prefix in each transmission; (ii) the devices are able to receive and decode all the transmissions in the downlink sub-carriers; and (iii) the devices report periodically to $B$ the channel quality of all downlink channels estimated based on the downlink channel pilots.
%
\begin{figure}[tb]
    \begin{center}
        \includegraphics[width=\linewidth]{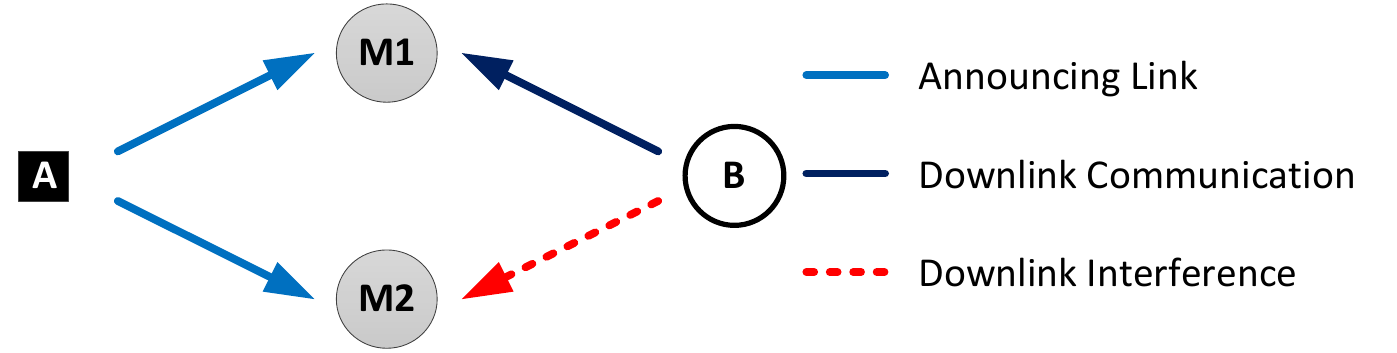}
    \end{center}
    \caption{Signal and interference links model Caption here}
    \label{fig:SignalInterferenceModel}
\end{figure}
\subsection{Signal Model} 
\label{sub:signal_model}

Assuming the announcer transmits in the $k^{th}$ downlink channel, as illustrated in Fig.~\ref{fig:SignalInterferenceModel}, then the composite signal $y_{k,i}$ received in the $k^{th}$ channel at the $i^{th}$ device is,
\begin{equation} \label{yk_withTwoSignals}
    y_{k,i} = h_{A,i} x_{A} + h_{B,i} x_{B} + Z
\end{equation}
where $h_{B,i}$ and $h_{A,i}$ denote respectively the complex channel gains between the $B$ and $A$ transmitters and the $i^{th}$ receiver, while $Z$ is the complex-valued Gaussian noise with variance $E[|Z|^2] = \sigma^2$. 
$x_{B}$ and $x_{A}$ are given by the circular zero-mean Gaussian complex signal transmitted by $A$ and $B$, where the respective variances are $E[|x_{A}|^2] = P_{A}$ and $E[|x_{B}|^2] = P_B$.
Where $P_A$ is assumed to be constant and dimensioned such that the monitors within the proximity region radius are able to detect the announcers transmission with very high probability.
As previously discussed, the downlink transmission is assumed to have a fixed rate $R_B$ that is to be met with zero-outage.
Then to met this requirement, $P_B$ is adapted via the power control scheme exposed in Section~\ref{sec:analysis}, which copes with the downlink channel conditions and the interference from the underlay announcer transmission.

On the downlink channels where only the downlink transmission is present, which corresponds to the orthogonal setting here considered as baseline, the received signal at the $i^{th}$ device is given by,
\begin{equation} \label{yj_withOneSignal}
    y_{j,i} = h_{B,i} x_{B} + Z, \mbox{ with } j \neq k
\end{equation}

All communication links are assumed to be characterized by block Rayleigh fading, i.e. the channel fading gains will not vary within a scheduling epoch.
$g_i = |h_i|^2$ is the channel envelope fading gain, which in the presence of Rayleigh fading follows an exponential distribution.
The associated Probability Density Function (PDF), $f_{g}(x)$, is given by,
\begin{equation}\label{expdistrib}
    f_{g}(x) = \frac{1}{\bar{g}} \exp \left(-\frac{x}{\bar{g}} \right) \mbox{, with } \bar{g} = 1
\end{equation}

We assume that $B$ has a fixed downlink transmission rate and is able to adapt its transmission power based on the Channel State Information available at the Transmitter (CSIT) of the $B-M_i$ links.
We further assume that the absence of any channel state information about the $A-M_i$ links, since prior to the announcer transmission it is assumed no prior communication between $A$ and $M_i$ has taken place.
On the other hand, it is assumed that $B$ knows the fixed transmission rate associated with the proximity announcement transmission, denoted as $R_A$.
All communications are performed using single-user point-to-point capacity-achieving Gaussian codebooks and the instantaneous achievable rate $R_{i}$ is given by the asymptotic Shannon capacity in AWGN, $R_{i} = \C{x} = W\log_2(1 + x)$, where $W$ is the channel bandwidth.
The target Signal-to-Noise Ratio (SNR) $\Gamma_x$, for a given target rate $R_x$, is defined as $\Gamma_x = C^{-1}(R_x) = 2^{R_x/W} - 1$.



\section{Analysis} 
\label{sec:analysis}

In our analysis, we focus on the downlink transmission power expenditure required to enable the underlay operation of the direct proximity discovery protocol.

\subsection{Underlay Fundamentals} 
\label{sub:underlay_fundamentals}

In the underlay proximity discovery setting, the monitor receives simultaneously the transmissions from $B$ and $A$.
Therefore, a two-user Gaussian Multiple Access Channel (MAC)~\cite{gamal2011network} is created at the $i^{th}$ monitor receiver $M_i$.
Denoting the rates of the signals present as $R_B$ and $R_A$, then the MAC defining inequalities are:
\begin{align}\label{eq:MACinequalities}
    R_A &\leq \C{\gamma_A} \nonumber \\ 
    R_B &\leq \C{\gamma_B} \nonumber \\ 
    R_A + R_B &\leq \C{\gamma_A + \gamma_B} 
\end{align}
\begin{figure}[tb]
    \begin{center}
        \includegraphics[width=\linewidth]{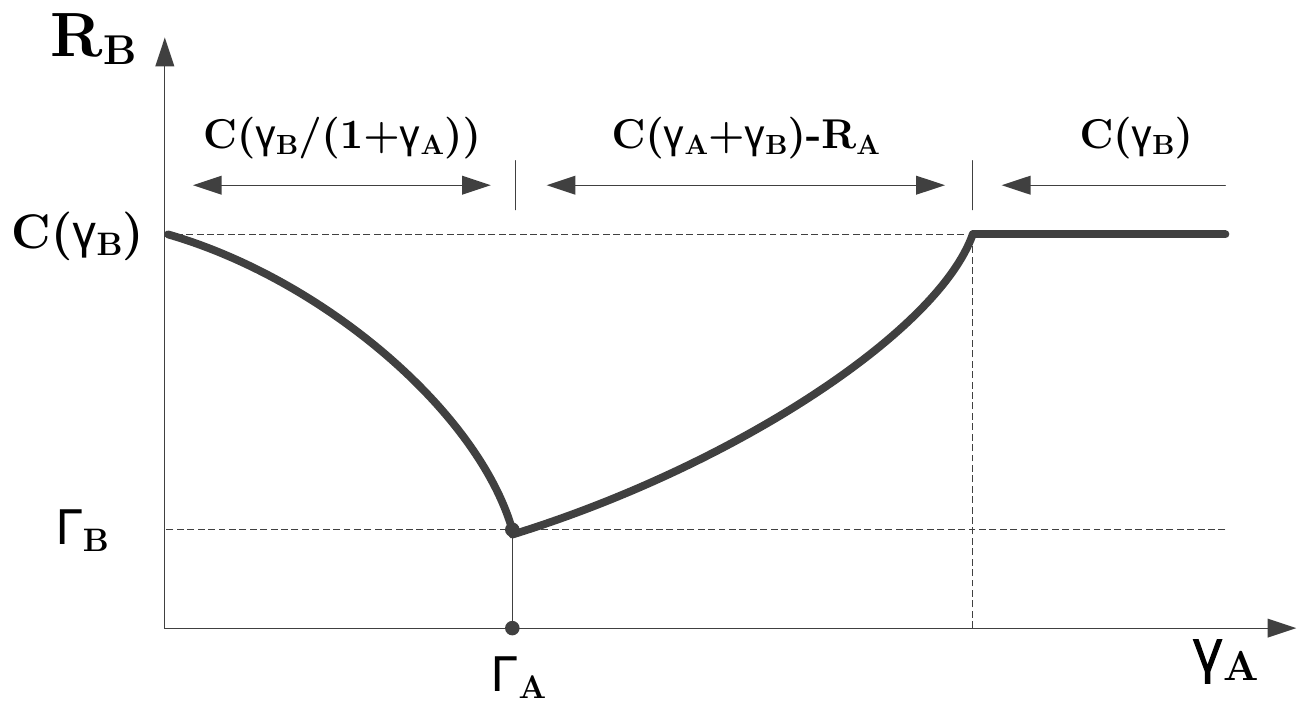}
    \end{center}
    \caption{The achievable $R_B$ when using Joint User Decoding (JD) given known $\gamma_B$ that does not put the $B-M_i$ link in outage, according with the channel realization $\gamma_A$ of the link $A-M_i$.}
    \label{fig:ZeroOutageRateBound}
\end{figure}
In our setting, we assume the announcer transmission to be fixed with rate $R_A$.
In the following for ease of exposition, we assume that $B$ can adapt its rate\footnote{
Later in Section~\ref{sub:underlay_power_control} we fix $B$'s rate and then cast the result derived in Lemma~\ref{prop:GammaB_JUD} to a power control setting that assumes downlink fixed rate.}. 
To illustrate how the maximum decodable $R_B$ is affected by the channel realization of the link $A-M_i$, consider the illustration in Fig.~\ref{fig:ZeroOutageRateBound}.
When the transmission from $A$ is not decodable, $\gamma_A < \Gamma_A$, the maximum decodable $R_B$ decreases, since a larger $\gamma_A$ corresponds to a larger non-decodable interference.
After $\gamma_A$ reaches $\Gamma_A$, the transmission from $A$ becomes decodable and $R_B$ starts to increase due to both signals being jointly decoded. Finally, when $\gamma_A$ becomes high enough, the downlink rate reaches its maximum possible value $R_B = C(\gamma_B)$.
In~\cite{7051287} this scenario as been considered and a crucial result for our setting derived that, when adapted to our context and notation, states the following:
\begin{lem}\label{prop:GammaB_JUD}
    Let there be a transmission with fixed rate $R_A$.
    Let $B$ know the rate $R_A$ and the SNR $\gamma_{B,i}$ at the $i^{th}$ monitor, but not $\gamma_{A}$.
    Then the maximal downlink transmission rate that is always decodable by $M_i$ is $R_B = C(\Gamma_B)$ where
    \begin{equation}\label{eq:GammaBSingleUser}
        \Gamma_B = \frac{\gamma_{B,i}}{1 + \Gamma_A}
    \end{equation}
    and $\Gamma_A=C^{-1}(R_A)$.  
\end{lem} 
\begin{proof}
    The proof, exposed in~\cite{7051287}, is based on finding the switching point between the region where the transmission from $B$ is decoded by treating the transmission from $A$ as noise, to the region where both signals are jointly decoded, as depicted in Fig.~\ref{fig:ZeroOutageRateBound}.
    The maximum downlink transmission $R_B = C(\Gamma_{B})$ then corresponds to that same switching point, which occurs when $\gamma_A = \Gamma_A$.
\end{proof}

In other words, when adapting the rate of the signal from $B$ to meet the defined upper bound, then $B$'s transmission can always be decoded, regardless if the transmission from $A$ is decoded. 
An interesting side effect of this rate upper bound, is on how it affects the outage of the announcing device transmission.
When we apply the downlink rate upper bound, all the MAC inequalities are respected, i.e. the signal from B is always decodable.
Then the decodability of $R_A$ will depend solely if $R_A \leq C(\gamma_A)$, i.e. the probability of the transmission from $A$ being decoded is given by $Pr\{\gamma_A \geq \Gamma_A\}$.
\begin{cor}\label{cor:GammaB_JUD}
    Let there be a transmission with rate $R_A$ and let $B$ know the rate $R_A$ and the SNR $\gamma_B$, and select its rate such that $R_B = C(\Gamma_B)$.
    Then the decodability of $R_A$ will depend solely on $R_A \leq C(\gamma_A)$.  
\end{cor} 
\begin{proof}
    From Lemma~\ref{prop:GammaB_JUD} it is stated that the transmission from B is always decodable.
    Therefore there are two decoding scenarios:
    (i) $B$ is decoded, by treating $A$ as noise, and subtracted from~\eqref{yk_withTwoSignals}, allowing $A$ to be decoded in the presence of noise if and only if $R_A \leq C(\Gamma_A)$;
    (ii) $A$ is jointly decoded with $B$, which is only true if all the inequalities in~\eqref{eq:MACinequalities} are met.
    The first inequality, $R_A \leq C(\gamma_A)$, is met for joint decoding to be possible.
    The second inequality is met, since $R_B \leq C(\Gamma_B) \leq C(\gamma_B)$.
    Finally, the third inequality is also met,
    \begin{align*}
        R_A + R_B &= C(\Gamma_A) + C(\Gamma_B) \\
                  &= C(\Gamma_A) + C(\gamma_B + \Gamma_A) - C(\Gamma_A) \\
                  &= C(\gamma_B + \Gamma_A) \leq C(\gamma_B + \gamma_A)
    \end{align*}
    where the last inequality comes from $\Gamma_A \leq \gamma_A$, one of the requirements for joint decoding to take place.
\end{proof}
Finally, the decoding here considered is based on a information-theoretic setup and codebooks.
Therefore, to bring this concept to practice one may resort to transmission techniques that are suitable for multiuser decoding.


\subsection{Underlay Downlink Power Control} 
\label{sub:underlay_power_control}

We now cast the zero-outage result provided in Lemma~\ref{prop:GammaB_JUD} to a power control setting, by assuming that each downlink transmission has fixed rate, given by $R_B$.

\subsubsection{Single Monitoring Device} 
\label{sub:single_monitoring_device}

We start by introducing the power control mechanism for the downlink transmission, assuming a single monitor within the proximity region.
Defining the instantaneous SNR of the downlink transmission
\begin{equation*}
    \gamma_B = \frac{P_B g d_B^{-\alpha}}{\sigma^2}
\end{equation*}
where $d_B^{-\alpha}$ is the path loss losses, $\sigma^2$ is the noise power, $g = |h|^2$ is the downlink channel gain realization and $P_B$ the downlink transmission power.
The asymptotic capacity on a Gaussian channel is then,
\begin{align*}
    R_B &=W \log_2 \left( 1 + \frac{\gamma_B}{1 + \Gamma_A}\right) \\ \nonumber
        &=W \log_2 \left( 1 + \frac{P_B g d_B^{-\alpha}}{\sigma^2} \cdot \frac{1}{1 + \Gamma_A}\right)
\end{align*}
where $W$ is the bandwidth of the assigned resources.
The term $\frac{1}{1 + \Gamma_A}$ accounts for the rate penalty due to the underlay announcer transmission, so to achieve zero-outage in the downlink as stated in Lemma~\ref{prop:GammaB_JUD}.  
The required power $P_B$ to meet the fixed rate $R_B$ is then given by,
\begin{equation}\label{eq:PowerInversion}
    P_B = \frac{1}{g} \cdot K \mbox{ where } K = \frac{\sigma^2 (1 + \Gamma_A)}{d_B^{-\alpha}} \cdot \left[ 2^{R_B/W} - 1\right].
\end{equation}
To cope with deep fades, a truncated channel inversion strategy is put in place~\cite{775366}.
In this setting, $B$ will only attempt transmission when the fade depth is above a cutoff threshold $\mu$, i.e. when $g > \mu$.
This link outage event is then,
\begin{equation}\label{eq:BaselineLinkOutage}
    Pr\{ g < \mu \} = \int_0^\mu f_g (x) dx = 1 -\exp(-\mu/\bar{g})
\end{equation}
where $\mu = -\bar{g} \log(1 - Pr\{ g < \mu \})$.
We note that although the downlink transmission will not occur in this event, the announcer transmission will still occur.
In other words, the announcer transmission is decodable with probability $Pr(\gamma_A \geq \Gamma_A)$.
The average transmit power required to sustain $R_B$, is then,
\begin{align}
    E[P_B] = \int_{\mu}^{\infty} \frac{1}{g} \cdot K f_g (g) d g_i = E_1\left(\frac{\mu}{\bar{g}}\right) \cdot K 
\end{align}
where $E_1(.)$ is the exponential integral function.

\subsubsection{Multiple Monitoring Devices} 
\label{sub:multiple_monitoring_devices}

When multiple monitoring devices are present within the proximity region, then the $P_B$ needs to be computed for the worst link.
This is the case since $R_B$ needs to be decodable at all the monitors within the proximity region.
The required power then becomes,
\begin{equation}
    P_B = \frac{1}{\min(g_1,...,g_N)} \cdot K
\end{equation}
where $g_i$ is the downlink channel gain of the $i^{th}$ monitor and where $K$ is assumed for analytical tractability to be the same for all monitors.
To account for the cutoff value we introduce the $n^{th}$ order statistics distribution:
\begin{equation*}
    f_n(x) = N f_g(x) \binom {N-1}{n-1} F(x)^{n-1} (1 - F(x))^{N-n}
\end{equation*}
from which we want $n = 1$, the distribution of the worse channel:
\begin{equation}
    f_1(x) = \frac{N}{\bar{g}} \exp\left(-\frac{N x}{\bar{g}}\right)
\end{equation}
The outage event is then defined as,
\begin{align}
    Pr\{ \min(g_1,...,g_N) < \mu \} &= \int_0^\mu f_1 (x) dx \\ \nonumber
                                    &= 1 -\exp(-N\mu/\bar{g})
\end{align}
and the associated cutoff threshold as,
\begin{equation}
    \mu = -\frac{\bar{g}}{N} \log \left(1 - Pr\{ \min(g_1,...,g_N) < \mu \} \right)
\end{equation}

The mean required $P_B$ in this setting is then given by,
\begin{equation}
    E[P_B | N] = E_1\left(\frac{N \mu}{\bar{g}}\right) \cdot N \cdot K
\end{equation}
%

\subsubsection{Multiple Channels} 
\label{sub:multiple_channels}

Now we assume that within the proximity region there are $N$ monitors, each with an associated downlink channel.
We further assume that each of these channels has different gains at each monitor.
In this setting, the channel chosen for the underlay operation should be the one that will require less expended power.

Assuming that out of the $N$ monitors the one with the worst channel gain in the $i^{th}$ channel is given by $u_i = \min(g_{1,i},...,g_{N,i})$, where $g_{n,i}$ corresponds to the downlink channel gain of the $i^{th}$ channel at the $n^{th}$ device.
Then, the channel that will lead to less power expenditures is the one with channel gain given by $u_{max} = \max\left(u_1,\cdots,u_N\right)$.
In other words, assuming the monitor with the worst channel conditions in each channel as the limiting factor, then the selection of the best channel in this conditions is modeled by the order statistics of the maximum value.
The corresponding density function, $f_N(x)$, is given by,
\begin{align}
    f_N(x) &= N f_1(x) F_1(x)^{N-1} 
\end{align}
and the cumulative function $F_N(x)$ given by,
\begin{equation}
    F_N(x) = \left[ 1 - \exp\left( - \frac{N x}{\bar{g}}\right) \right]^M
\end{equation}

The outage event is then defined by,
\begin{align}\label{eq:outageWorseChannel}
    Pr\{ u_{max} < \mu \} = F_N(\mu) = \left( 1 -  \exp\left( -\frac{N \mu}{\bar{g}}  \right) \right)^N
\end{align}
where
\begin{equation}\label{eq:cutoffWorseChannel}
    \mu = -\frac{\bar{g}}{N} \log(1 - Pr\{ u_{max} < \mu \}^{1/N})
\end{equation}

In these conditions the average power used is given by,
\begin{equation}
    E[P_B | N] = K \int_{\mu}^{\infty} \frac{1}{x} f_N(x) dx
\end{equation}
which can be computed via numerical integration.

The sum-rate per used resource, $S$, focused on the worst channel, is given by,
\begin{equation}\label{eq:SumRateUnderlay}
    S = R_B \cdot Pr\{ u_{max} \geq \mu \} + R_A \cdot Pr\{\gamma_A \geq \Gamma_A\}.
\end{equation}
As measure of energy efficiency we consider the amount of energy required to transmit a bit of information, $\psi$, given by,
\begin{align}
    \psi &= E\left[\frac{P_B}{S}\right] = \frac{K}{S} \int_{\mu}^{\infty} \frac{1}{x} f_N (x) d x
\end{align}
which can be computed via numerical integration.



\subsection{Orthogonal Downlink Power Control} 
\label{sub:baseline}

We now consider as baseline the setting where the announcer and downlink transmissions are allocated orthogonal resources by the network.
Starting from the asymptotic capacity on a gaussian channel at the $i^{th}$ device,
\begin{equation}
    R_B = W \log_2 \left( 1 + \frac{P_B g_i d_B^{-\alpha}}{\sigma^2}\right)
\end{equation}
the required downlink power $P_B$ to met the fixed rate $R_B$ is then given by,
\begin{equation}\label{eq:BaselineChannelInversion}
    P_B = \frac{1}{g_i} \cdot \frac{\sigma^2 \left[ 2^{R_B/W} - 1\right]}{r^{-\alpha}}.
\end{equation}
To provide a fair comparison we focus on the power required to serve the worst downlink channel.
In this setting, the link outage event is provided by~\eqref{eq:outageWorseChannel} and the associated cutoff threshold by~\eqref{eq:cutoffWorseChannel}.
The average transmit power required to sustain $R_B$, is then,
\begin{align}
    E[P_B] &= \int_{\mu}^{\infty} \frac{1}{x} \cdot \frac{\sigma^2 \left[ 2^{R_B/W} - 1\right]}{d_B^{-\alpha}} f_N(x) d x
\end{align}
which can be computed numerically.

The sum-rate per used resource $S$, from the worst channel and the channel dedicated for the announcing device transmission, is given by,
\begin{equation}\label{eq:SumRateOrthogonal}
    S = \frac{R_B \cdot Pr\{ u_{max} \geq \mu \} + R_A \cdot Pr\{\gamma_A \geq \Gamma_A\}}{2}
\end{equation}
and the associated energy efficiency, $\psi$, given by,
\begin{align}
    \psi &= E\left[\frac{P_B}{S}\right] = \frac{1}{S} \int_{\mu}^{\infty} \frac{1}{x} \cdot \frac{\sigma^2 \left[ 2^{R_B/W} - 1\right]}{d_B^{-\alpha}} f_N(x) d x.
\end{align}
%

\section{Numerical Results and Discussion} 
\label{sec:numerical_results_and_discussion}
\begin{table}[t]
    \centering
        \begin{tabular}{ l c l c }
        \hline      
        \textbf{Parameter} & \textbf{Value} & \textbf{Parameter} & \textbf{Value}\\ \hline
        $W$ & 180 [kHz] & $T_A$ & 5 [ms] \\
        $d_B$ & 200 [m] & $d_A$ & 20  [m] \\
        $P_A$ & 20 [dBm] & $\sigma^2$ & -97 [dBm] \\
        $\alpha$ & 4 & $Pr\{\gamma_A \geq \Gamma_A\}$ & 0.99 \\
        $R_B$ & 5 $[B/s^2]$ & $Pr\{ u_{max} \geq \mu \}$ & 0.99 \\
        $N$ & 20 & A's Payload & $[100,1100]$ [B]\\
        \hline
        \end{tabular}
    \caption{Simulation scenario settings.}
    \label{tab:SimulationScenarioSettings}
\end{table}
\begin{figure}[t]
    \begin{center}
        \includegraphics[width=\linewidth]{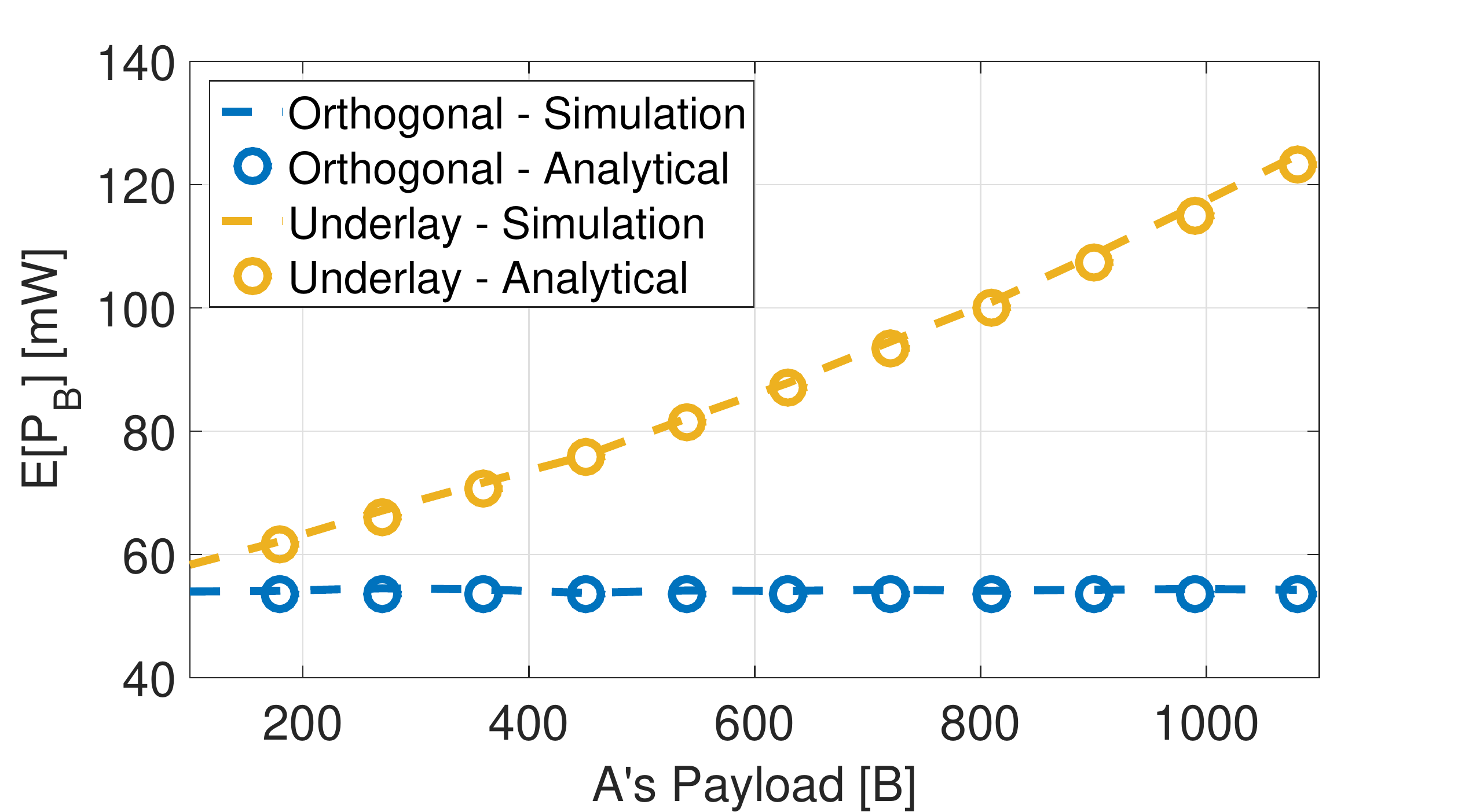}
    \end{center}
    \caption{Base station power consumption comparison between underlay and orthogonal discovery, for multiple monitoring devices.}
    \label{fig:PowerConsumption}
\end{figure}
In this section we present the numerical results corresponding to the performance of the proximity discovery in the underlay setting, in terms of expended power $E[P_B]$ and energy efficiency $\psi$ versus the announcer payload (computed as $R_A \cdot W \cdot T_A$, where $T_A$ is the duration of the announcer transmission).
These numerical results were obtained by evaluating the corresponding analytical expressions, derived in the previous section and through stochastic simulations.
Table~\ref{tab:SimulationScenarioSettings} lists the relevant assumed system parameters.
We first observe, that the sum-rate achieved in the underlay setting will always be higher than in the orthogonal resources setting, as can been seen when comparing~\eqref{eq:SumRateUnderlay} with~\eqref{eq:SumRateOrthogonal}.
This higher sum-rate comes at a cost of extra downlink transmission expenditure when compared to the orthogonal resources setting, as shown in Fig.~\ref{fig:PowerConsumption}. Specifically, in the underlay setting the higher is the announcers rate (and corresponding payload) the higher will the required downlink transmission power be.
Meanwhile from an energy efficiency perspective, in Fig.~\ref{fig:EnergyEfficiency} we observe that for lower announcer payloads and despite the higher downlink power required, the underlay scheme is more energy efficient.
This is the case, since the penalty introduced by the term $(1 + \Gamma_A)$ in~\eqref{eq:PowerInversion} is for low payloads low enough to not offset significantly the sum-rate gains in the underlay setting.


\section{Conclusion}
\label{sec:Conclusion}

Proximity discovery is an essential enabler for D2D communications.
We have shown that reliable underlay proximity discovery is feasible, achieving higher sum-rates than in the orthogonal resources setting, at the cost of an higher amount of expended downlink power.
On the other hand, we have observed that for low announcer payloads the underlay discovery is up to twice more energy efficient than the orthogonal setting.
In terms of future work direction, the proposed underlay proximity discovery should be integrated in a system level scheduler and extended to a multi-antenna setting.
\begin{figure}[t]
    \begin{center}
        \includegraphics[width=\linewidth]{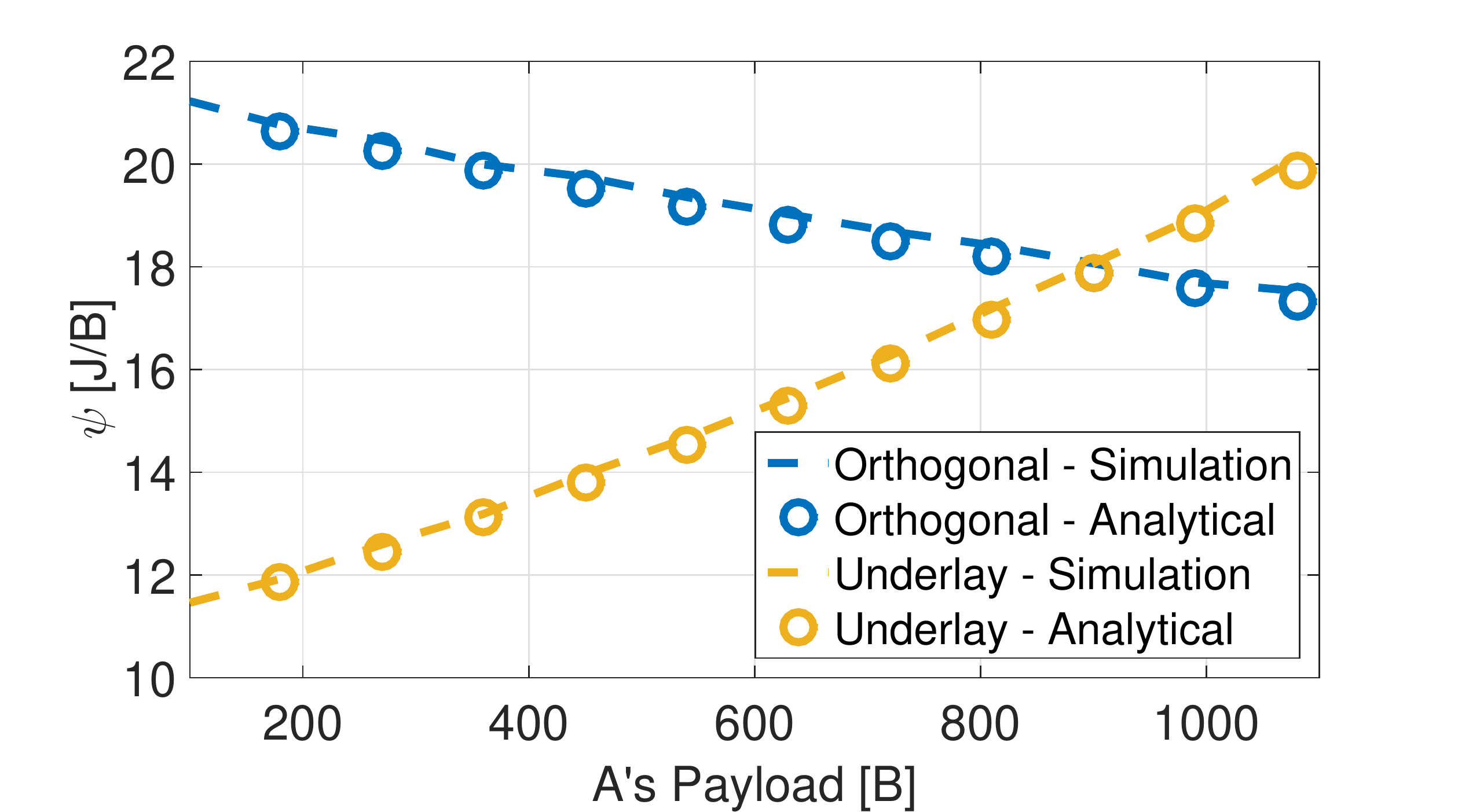}
    \end{center}
    \caption{Energy efficiency comparison between underlay and orthogonal discovery, for multiple monitoring devices.}
    \label{fig:EnergyEfficiency}
\end{figure}
\section*{Acknowledgment}
Part of this work has been performed in the framework of the Horizon 2020 project FANTASTIC-5G (ICT-671660), which is partly funded by the European Union. The authors would like to acknowledge the contributions of their colleagues in FANTASTIC-5G. 

\bibliographystyle{ieeetr}

\end{document}